\documentclass[11pt]{article}%
\usepackage{amsfonts}
\usepackage{amsmath}
\usepackage{amssymb}
\usepackage{graphicx}%
\providecommand{\U}[1]{\protect\rule{.1in}{.1in}}
\newtheorem{theorem}{Theorem}
\newtheorem{acknowledgement}[theorem]{Acknowledgement}

\newtheorem{notation}[theorem]{Notation}

\newtheorem{proposition}[theorem]{Proposition}

\newenvironment{proof}[1][Proof]{\noindent\textbf{#1.} }{\ \rule{0.5em}{0.5em}}
\begin{document}

\title{Phase Space Representation of the Density Operator: Bopp Pseudodifferential
Calculus and Moyal Product}
\author{Maurice de Gosson\\Austrian Academy of Sciences\\Acoustics Research Institute\\1010, Vienna, AUSTRIA}
\maketitle

\begin{abstract}
Bopp shifts, introduced in 1956, played a pivotal role in the statistical
interpretation of quantum mechanics. As demonstrated in our previous work,
Bopp's construction provides a phase-space perspective of quantum mechanics
that is closely connected to the Moyal star product and its role in
deformation quantization. In this paper, we both review and expand on our
exploration of Bopp quantization, emphasizing its relationship with the Moyal
product and its applications in elementary deformation quantization. Notably,
we apply these constructions to the density operator, which represents mixed
states in quantum mechanics, offering novel insights into its role within this framework.

\end{abstract}

\section{Introduction}

In a series of papers \cite{gobopp1,gobopp2,Boppgolu1,Boppgolu2,Boppgolu3}
(also see \cite{Birkbis}, Chapters 18 and 19) we developed a rigorous theory
of phase space quantization which we named \textquotedblleft Bopp
quantization\textquotedblright\ in honor of the Germane mathematician Fritz
Bopp who was the first to consider this kind of quantization in \cite{Bopp}.
Bopp suggested to replace the standard quantization rules $\widehat{x}=x$and
$\widehat{p}=-i\hbar\partial_{x}$ with operators $\widetilde{x}$ and
$\widetilde{p}$ acting on phase space function and defined by
\begin{equation}
\widetilde{x}=x+\frac{i\hbar}{2}\partial_{p}\text{ \ }\widetilde{p}%
=p-\frac{i\hbar}{2}\partial_{x}\label{Bopp1}%
\end{equation}
where $\partial_{x}$ and $\partial_{p}$ are the gradient operators in the
position and momentum variables $x=(x_{1},...,x_{n})$ and $p=(p_{1}%
,...,p_{n})$. The \textquotedblleft Bopp shifts\textquotedblright, as they are
sometimes called, satisfy the same canonical commutation relations%
\[
\lbrack\widetilde{x}_{j},\widetilde{p}_{k}]=i\hbar\delta_{jk}%
\]
as $\widehat{x_{j}}$ and $\widehat{p_{k}}$. Bopp's rules where heuristic but
they can be justified by noting that they intertwine the standard rules by the
cross-Wigner transform%
\begin{equation}
W(\psi,\phi)(x,p)=\left(  \tfrac{1}{2\pi\hbar}\right)  ^{n}\int e^{--\frac
{i}{\hbar}py}\psi(x+\tfrac{1}{2}y)\phi^{\ast}(x-\tfrac{1}{2}%
y)dy.\label{crosswig}%
\end{equation}
In fact, a straightforward calculation shows that
\begin{equation}
W(\widehat{x}\psi,\phi)=\widetilde{x}W(\psi,\phi)\text{ },\text{
}W(\widehat{p}\psi,\phi)=\widetilde{p}W(\psi,\phi)\label{BoppWig}%
\end{equation}
for all $\psi,\phi$ or, somewhat more explicitly,%
\begin{align}
W(x_{j}\psi,\phi) &  =(x_{j}+\frac{i\hbar}{2}\partial_{p_{j}})W(\psi
,\phi)\label{BoppWig1}\\
W(-i\hbar\partial_{x_{j}}\psi,\phi) &  =(p_{j}-\frac{i\hbar}{2}\partial
_{x_{j}})W(\psi,\phi)\label{BoppWig2}%
\end{align}
for $1\leq j\leq n$. These relations motivate the definition of a
pseudodifferential on phase space, associating to a symbol (or observable)
$a=a(x,p)$ belonging to some suitable function space an operator%
\[
\widetilde{A}:\mathcal{S}(\mathbb{R}^{2n})\longrightarrow\mathcal{S}^{\prime
}(\mathbb{R}^{2n})
\]
with symbol $\widetilde{a}\in(\mathbb{R}^{2n}\times\mathbb{R}^{2n})$ defined
by
\begin{equation}
\widetilde{a}(x,p;\zeta_{x},\zeta_{p})=a(x-\tfrac{1}{2}\zeta_{p},p+\tfrac
{1}{2}\zeta_{x}).\label{rev1}%
\end{equation}
This justifies the formal nitration%
\[
\widetilde{A}=\operatorname*{Op}\nolimits_{\mathrm{Bopp}}(a)=a(x+\frac{i\hbar
}{2}\partial_{p},p-\frac{i\hbar}{2}\partial_{x}).
\]
A fundamental fact is the the \textquotedblleft Bopp
operator\textquotedblright\ $\operatorname*{Op}\nolimits_{\mathrm{Bopp}}(a)$
is relate to the usual Weyl \ operator $\widehat{A}=\operatorname*{Op}%
_{\mathrm{Weyl}}(a)$ by an infinity of intertwining relations
\begin{equation}
\operatorname*{Op}\nolimits_{\mathrm{Bopp}}(a)U_{\phi}=U_{\phi}%
\operatorname*{Op}\nolimits_{\mathrm{Weyl}}(a)\label{rev2}%
\end{equation}
where the \textquotedblleft wavepacket transform\textquotedblright\ $U_{\phi}%
$, defined by
\begin{equation}
U_{\phi}\psi=(2\pi\hbar)^{n/2}W(\psi,\phi)\label{rev4}%
\end{equation}
is a partial isometry $L^{2}(\mathbb{R}^{n})\longrightarrow L^{2}%
(\mathbb{R}^{2n})$. This property allows to express most properties of
$\operatorname*{Op}\nolimits_{\mathrm{Bopp}}$ using those of
$\operatorname*{Op}\nolimits_{\mathrm{Weyl}}(a)$. It turns out that Bopp
operators are closely related to the notion of Moyal starproduct
$\bigstar_{\hbar}$. In fact, the action of $\operatorname*{Op}%
\nolimits_{\mathrm{Bopp}}(a)$ on a function $\Psi\in\mathcal{S}(\mathbb{R}%
^{2n})$ is given by the remarkable relation%
\begin{equation}
\operatorname*{Op}\nolimits_{\mathrm{Bopp}}(a)\Psi=a\bigstar_{\hbar}%
\Psi.\label{rev5}%
\end{equation}

In the present work we will specialize to the case where $\operatorname*{Op}%
\nolimits_{\mathrm{Weyl}}(a)$ is a density operator. Such operators are just
positive semidefinite operators with trace one on $L^{2}(\mathbb{R}^{n})$ and
represent the mixed states of quantum mechanics, where they play a central
role. More precisely:

\begin{itemize}
\item In Section 2 we review the basic properties of Bopp pseudodifferential
operators and their relation with Weyl calculus and the Moyal product;

\item In Section 3 we apply these notions to density operators, and study the
properties of $\widetilde{\rho}$ when $\widehat{\rho}$ is a density operator
on $L^{2}(\mathbb{R}^{n})$. We will see that while $\widetilde{\rho}$ is not a
density operator, its restriction to the Hilbert space $\operatorname{Im}%
U_{\phi}\subset L^{2}(\mathbb{R}^{2n})$ is;

\item In Section 4 we apply the results above to the deformation quantization
of density operators, this is made possible using the Moyal product;

\item We discuss some possible extensions of our results in Section 5 where we
also collect some references to related works.
\end{itemize}

We take in this work a minimalist approach to make the paper as self-contained
as possible. For instance, the star product on Hilbert spaces has been studied
by Dito \cite{Dito1} but we do not recourse to this theory explicitly.

\begin{notation}
We denote by $\mathbb{R}^{2n}\equiv\mathbb{R}_{x}^{n}\times\mathbb{R}_{p}^{n}$
the phase space; it comes equipped with the standard symplectic form
$\sigma(z,z^{\prime})=Jz\cdot z^{\prime}=(z^{\prime})^{T}Jz$ where
\[
J=%
\begin{pmatrix}
0_{n\times n} & I_{n\times n}\\
-I_{n\times n} & 0_{n\times n}%
\end{pmatrix}
\]
is the standard symplectic matrix. The symplectic group associated with
$\omega$ is denoted by $\operatorname*{Sp}(n)$; it consists of all linear
automorphisms $S$ of phase space such that $\omega(Sz,Sz^{\prime}%
)=\omega(z,z^{\prime})$ for all $z,z^{\prime}$ in $\mathbb{R}_{z}^{2n}$;
equivalently $S^{T}JS=SJS^{T}=J$. The scalar product in of $x,y\in
\mathbb{R}^{m}$ is denoted by $(x,y)\longmapsto xy$. The scalar product on the
Hilbert space $L^{2}(\mathbb{R}^{m})$ is%
\[
(f,g)\longmapsto\int f(u)g^{\ast}(u)du
\]
where  $g^{\ast}$ is the complex conjugate of $g$.
\end{notation}

\section{Weyl and Bopp Pseudodifferential Operators}

For complete proofs and details the reader can consult Chapters 18 and 19 of
\cite{Birkbis} or Chapter 11 in \cite{WIGNER}, and Folland \cite{Folland}.

\subsection{Weyl calculus and the Moyal star product}

Traditionally (in particularly in elementary texts) the Weyl operator
$\widehat{A}=\operatorname*{Op}_{\mathrm{Weyl}}(a)$ is defined, for $\psi
\in\mathcal{S}(\mathbb{R}^{n})$, by%
\begin{equation}
\widehat{A}\psi(x)=\left(  \tfrac{1}{2\pi\hbar}\right)  ^{n}\int e^{\frac
{i}{\hbar}p(x-y)}a(\tfrac{1}{2}(x+y),p)\psi(y)dpdy \label{Weyl1}%
\end{equation}
associating to a linear operator $\widehat{A}$ to the symbol $a$ (belong to
some suitable function space, say $\mathcal{S}(\mathbb{R}^{2n})$ to be on the
safe side). It highlights the relation between the symbol $a$ and the kernel
$\tfrac{1}{2}$ of a Weyl operator:%
\begin{align}
K_{\widehat{A}}(x,y)  &  =\left(  \tfrac{1}{2\pi\hbar}\right)  ^{n}\int
e^{\frac{i}{\hbar}p(x-y)}a(\frac{1}{2}(x+y),p)\psi(y)dp\label{ker1}\\
a(x,p)  &  =\int e^{-\frac{i}{\hbar}py}K_{\widehat{A}}(x+\tfrac{1}%
{2}y,x-\tfrac{1}{2}y)dy. \label{ker2}%
\end{align}
More interesting and useful in our context is to use the harmonic
representation%
\begin{equation}
\widehat{A}\psi=\left(  \tfrac{1}{2\pi\hbar}\right)  ^{n}\int a_{\sigma
}(z)\widehat{T}(z)\psi dz \label{Weyl2}%
\end{equation}
where $a_{\sigma}=F_{\sigma}a$ is the symplectic Fourier transform of $a$ (see
Appendix B in \cite{WIGNER}):%
\begin{equation}
F_{\sigma}a(z)=\left(  \tfrac{1}{2\pi\hbar}\right)  ^{n}\int e^{-\frac
{i}{\hbar}\sigma(z,z^{\prime})}a(z^{\prime})dz^{\prime} \label{SAFT}%
\end{equation}
and $\widehat{T}(z)$ the Heisenberg--Weyl displacement operator:%
\begin{equation}
\widehat{T}(z_{0})\psi(x)=e^{\frac{i}{\hbar}(p_{0}x-\frac{1}{2}p_{0}x_{0}%
)}\psi(x-x_{0}). \label{HW}%
\end{equation}
One of the advantages of the harmonic representation is that it allows to
describe in a neat way the composition rule of two Weyl operators. Suppose
that $\widehat{A}=\operatorname*{Op}_{\mathrm{Weyl}}(a)$ and $\widehat{B}%
=\operatorname*{Op}_{\mathrm{Weyl}}(b)$ are such that $\widehat{A}\widehat{B}$
exists (this is the case if, for example, $\widehat{B}:\mathcal{S}%
(\mathbb{R}^{n})\longrightarrow\mathcal{S}(\mathbb{R}^{n})$). Then
$\widehat{A}\widehat{B}=\widehat{C}=\operatorname*{Op}_{\mathrm{Weyl}}(c)$
where $c$ is the \emph{Moyal star product \cite{Groenewold,Moyal} }%
$a\bigstar_{\hbar}b$, defined by
\begin{equation}
a\bigstar_{\hbar}b(z)=\left(  \tfrac{1}{4\pi\hbar}\right)  ^{2n}\int
e^{\frac{i}{2\hbar}\sigma(z^{\prime},z^{\prime\prime})}a(z+\tfrac{1}%
{2}z^{\prime})b(z-\tfrac{1}{2}z^{\prime\prime})dz^{\prime}dz^{\prime\prime}
\label{Moyal1}%
\end{equation}
or, equivalently,%
\begin{equation}
a\bigstar_{\hbar}b(z)=\left(  \tfrac{1}{\pi\hbar}\right)  ^{2n}\int
e^{-\frac{2i}{\hbar}\sigma(u-z,v-z)}a(u)b(v)dudv. \label{Moyal2}%
\end{equation}

A remarkable property of the Moyal product is its \textquotedblleft phase
space cyclicity\textquotedblright:
\begin{equation}
\int a\bigstar_{\hbar}b(z)dz=\int a(z)b(z)dz=\int b\bigstar_{\hbar}a(z)dz.
\label{cycle}%
\end{equation}
To prove this it suffices to note that in view of formula (\ref{Moyal2})%
\begin{align*}
\int a\bigstar_{\hbar}b(z)dz  &  =\left(  \tfrac{1}{\pi\hbar}\right)
^{2n}\int\left(  e^{\frac{2i}{\hbar}\sigma(u-z,v-z)}dz\right)  a(u)b(v)dudv\\
&  =\int\left(  \int e^{-\frac{2i}{2\hbar}\sigma(u-v,z)}dz\right)
e^{-\frac{2i}{\hbar}\sigma(u,v)}a(u)b(v)dudv\\
&  =\int\delta(u-v)e^{-\frac{2i}{\hbar}\sigma(u,v)}a(u)b(v)dudv\\
&  =\int a(u)b(v)du.
\end{align*}
One shows that the Moyal product is bounded on $L^{1}(\mathbb{R}^{2n})$; in
fact%
\begin{equation}
|a\bigstar_{\hbar}b(z)|\leq||a||_{L^{1}}||b||_{L^{1}}. \label{bound}%
\end{equation}

As an example let us calculate the starproduct of two cross-Wigner transform:

\begin{proposition}
\label{Propwigmo}Let $\psi,\psi^{\prime},\phi,\phi^{\prime}$ be in
$L^{2}(\mathbb{R}^{n})$. Then $W(\psi,\phi)\bigstar_{\hbar}W(\psi^{\prime
},\phi^{\prime})\in L^{2}(\mathbb{R}^{2n})$ and we have%
\begin{equation}
W(\psi,\phi)\bigstar_{\hbar}W(\psi^{\prime},\phi^{\prime})=\left(  \tfrac
{1}{2\pi\hbar}\right)  ^{n}(\psi^{\prime}|\phi))_{L^{2}(\mathbb{R}^{n})}%
W(\psi,\phi^{\prime}). \label{wigwig}%
\end{equation}

\end{proposition}

\begin{proof}
Recall that \cite{Birkbis,WIGNER} $W(\psi,\phi)$ (\textit{resp}.$W(\psi
^{\prime},\phi^{\prime})$) is the Weyl symbol of the operator $\widehat{A}$
(\textit{resp}. $\widehat{A}^{\prime}$) with kernel $K=(2\pi\hbar)^{-n}%
\psi\otimes\phi^{\ast}$(\textit{resp}.$K^{\prime}=(2\pi\hbar)^{-n}\psi
^{\prime}\otimes\phi^{\prime\ast}$). It follows that $W(\psi,\phi
)\bigstar_{\hbar}W(\psi^{\prime},\phi^{\prime})$ is the Weyl symbol of
$\widehat{B}=\widehat{A}\widehat{A}^{\prime}$ whose kernel is
\begin{align*}
L(x,y)  &  =\int K(x,u)K^{\prime}(u,y)du\\
&  =\left(  \tfrac{1}{2\pi\hbar}\right)  ^{2n}\psi(x)\phi^{\prime\ast}%
(y)\int\psi^{\prime}(u)\phi^{\ast}(u)du\\
&  =\left(  \tfrac{1}{2\pi\hbar}\right)  ^{2n}\psi(x)\phi^{\prime\ast}%
(y)(\psi^{\prime}|\phi))_{L^{2}(\mathbb{R}^{n})}.
\end{align*}
The starproduct $W(\psi,\phi)\bigstar_{\hbar}W(\psi^{\prime},\phi^{\prime})$
is the Weyl symbol of the operators with kernel $L$, both are related by
formula (\ref{ker2}), that is
\begin{align*}
W(\psi,\phi)\bigstar_{\hbar}W(\psi^{\prime},\phi^{\prime})(z)  &  =\int
e^{-\frac{i}{\hbar}py}L(x+\tfrac{1}{2}y,x-\tfrac{1}{2}y)dy\\
&  =\left(  \tfrac{1}{2\pi\hbar}\right)  ^{2n}(\psi^{\prime}|\phi
))_{L^{2}(\mathbb{R}^{n})}\int e^{-\frac{i}{\hbar}py}\psi(x+\tfrac{1}{2}%
y)\phi^{\prime\ast}(x-\tfrac{1}{2}y)dy\\
&  =\left(  \tfrac{1}{2\pi\hbar}\right)  ^{n}(\psi^{\prime}|\phi
))_{L^{2}(\mathbb{R}^{n})}W(\psi,\phi^{\prime})(z).
\end{align*}

\end{proof}

\subsection{Definition of Bopp operators}

Let us replace the symbol $b$ in the Moyal $a\bigstar_{\hbar}b$ with $\Psi
\in\mathcal{S}(\mathbb{R}^{2n})$ (it will play the role of a \textquotedblleft
phase space function\textquotedblright); the formula (\ref{Moyal2}) reads
\cite{Folland,Birkbis}%
\begin{equation}
a\bigstar_{\hbar}\Psi(z)=\left(  \tfrac{1}{\pi\hbar}\right)  ^{2n}\int
e^{-\frac{i}{2\hbar}\sigma(u-z,v-z)}a(u)\Psi(v)dudv.
\end{equation}
Performing the change of variables $v=z-\frac{1}{2}z_{0}$ this becomes, after
some calculations,%
\begin{equation}
a\bigstar_{\hbar}\Psi(z)=\left(  \tfrac{1}{2\pi\hbar}\right)  ^{n}\int
a_{\sigma}(z_{0})e^{-\frac{i}{\hbar}\sigma(z,z_{0})}\Psi(z-\frac{1}{2}%
z_{0})dz_{0}. \label{apsiosi}%
\end{equation}
We now define the operator $\widetilde{A}=\operatorname*{Op}_{\mathrm{Bopp}%
}(a)$ by
\begin{equation}
\widetilde{A}\Psi=a\bigstar_{\hbar}\Psi,\text{ }\Psi\in\mathcal{S}%
(\mathbb{R}^{n}). \label{defBopp}%
\end{equation}
Introducing the \textquotedblleft Bopp displacement operator\textquotedblright%
\[
\widetilde{T}(z_{0})\Psi(z)=^{-\frac{i}{\hbar}\sigma(z,z_{0})}\Psi(z-\frac
{1}{2}z_{0})
\]
this \ definition can be rewritten
\begin{equation}
\widetilde{A}\Psi=\left(  \tfrac{1}{2\pi\hbar}\right)  ^{n}\int a_{\sigma
}(z_{0})\widetilde{T}(z_{0})\Psi dz_{0}. \label{aharmonic}%
\end{equation}
The linear operator $\widetilde{A}:\mathcal{S}(\mathbb{R}^{2n})\longrightarrow
\mathcal{S}^{\prime}(\mathbb{R}^{2n})$ defined by the harmonic representation
above is called the \emph{Bopp operator} with symbol $a$ and denote it by
$\widetilde{A}=\operatorname*{Op}_{\mathrm{Bopp}}(a)$ )observe the formal
similarity of this definition with $\operatorname*{Op}_{\mathrm{Weyl}}(a)$
given by (\ref{Weyl2}).

The operator $\widetilde{A}$ can also be viewed as a Weyl operator
$\mathcal{S}(\mathbb{R}^{2n})\longrightarrow\mathcal{S}^{\prime}%
(\mathbb{R}^{2n})$ \ in view of Schwartz's kernel theorem; explicitly:

\begin{proposition}
\label{prop1}We have $\widetilde{A}$ $=\operatorname*{Op}_{\mathrm{Weyl}%
}(\widetilde{a})$ where the symbol $\widetilde{a}\in\mathcal{S}^{\prime
}(\mathbb{R}^{2n})$ is given by
\begin{equation}
\widetilde{a}(z,\zeta)=a(z-\tfrac{1}{2}J\zeta)\text{ \ },\text{ \ }%
z=(x,p),\zeta=(\zeta_{x},\zeta_{p}) \label{Boppsymbol}%
\end{equation}
where $J=%
\begin{pmatrix}
0 & I\\
-I & 0
\end{pmatrix}
\in\operatorname*{Sp}(n)$ is the standard symplectic matrix.
\end{proposition}

\begin{proof}
See \cite{Birkbis} Thm. 4.2 or \cite{WIGNER} Prop. 44.
\end{proof}

Formula (\ref{Boppsymbol}) can be somewhat more explicitly be written%
\[
\widetilde{a}(x,p;\zeta_{x},\zeta_{p})=a(x-\tfrac{1}{2}\zeta_{p},p+\tfrac
{1}{2}\zeta_{x})
\]
justifying a posteriori the formal notation
\[
\widetilde{A}=a(x+\frac{i\hbar}{2}\partial_{p},p-\frac{i\hbar}{2}\partial
_{x})
\]
as mentioned in the Introduction. Notice that formulas (\ref{BoppWig1}) and
(\ref{BoppWig2}) can be rewritten
\begin{align}
x_{j}\bigstar_{\hbar}W(\psi,\phi)  &  =(x_{j}+\frac{i\hbar}{2}\partial_{p_{j}%
})W(\psi,\phi)\label{BoppWig3}\\
p_{j}\bigstar_{\hbar}W(\psi,\phi)  &  =(p_{j}-\frac{i\hbar}{2}\partial_{x_{j}%
})W(\psi,\phi). \label{BoppWig4}%
\end{align}

The Weyl pseudodifferential calculus enjoys an important symmetry property,
that of symplectic covariance. That is, for every symplectic automorphism
$S\in\operatorname*{Sp}(n)$ we have
\begin{equation}
\operatorname*{Op}\nolimits_{\mathrm{Weyl}}(a\circ S^{-1})=\widehat{S}%
\operatorname*{Op}\nolimits_{\mathrm{Weyl}}(a)\widehat{S}^{-1} \label{sympco}%
\end{equation}
where $\widehat{S}\in\operatorname*{Mp}(n)$ is anyone of the two metaplectic
operators covering $S$ (recall that the metaplectic group $\operatorname*{Mp}%
(n)$ is the unitary representation of the double covering of the symplectic
group \cite{Birk}). \ This formula shows that Weyl quantizations with respect
to different symplectic frames are isomorphically equivalent. Voros
\cite{Wong2} has proven that this property is characteristic of the Weyl
calculus. The property of symplectic covariance carries over to Bopp operators
since they are Weyl operators in view of \ the proposition above. One has to
replace the groups $\operatorname*{Sp}(n)$ and $\operatorname*{Mp}(n)$ with
$\operatorname*{Sp}(2n)$ and $\operatorname*{Mp}(2n)$. In particular:

\begin{proposition}
Let $S\in\operatorname*{Sp}(n)$. We have%
\begin{equation}
\operatorname*{Op}\nolimits_{\mathrm{Bopp}}(a\circ S^{-1})=\widetilde{S}%
\operatorname*{Op}\nolimits_{\mathrm{Bopp}}(a)\widetilde{S}^{-1}
\label{opbopps}%
\end{equation}
where $\widetilde{S}\in\operatorname*{Mp}(2n)$ is defined by \ $\widetilde{S}%
\Psi(z)=\Psi(Sz)$ and $\widetilde{S}^{-1}\Psi(z)=\Psi(S^{-1}z)$.

\begin{proof}
Let $\widetilde{a}(z,\zeta)=a(z-\tfrac{1}{2}J\zeta)$ .and $M_{S}=S^{-1}\oplus
S^{T}\in\operatorname*{Sp}(2n)$ where $S\in\operatorname*{Sp}(n)$. We have We
have
\begin{align*}
\widetilde{a}(M_{S}(z,\zeta))  &  =\widetilde{a}(S^{-1}z,S^{T}\zeta)\\
&  =a(S^{-1}z-\tfrac{1}{2}JS^{T}\zeta=\\
&  =a(S^{-1}(z-\tfrac{1}{2}J\zeta))
\end{align*}
last equality because $JS^{T}=S^{-1}J$ since $S\in\operatorname*{Sp}(n)$. It
follows by the property of symplectic covariance that
\[
\operatorname*{Op}\nolimits_{\mathrm{Bopp}}(a\circ S^{-1})=\widetilde{S}%
\operatorname*{Op}\nolimits_{\mathrm{Bopp}}(a)\widetilde{S}^{-1}%
\]
where \ $=\widetilde{S}$. The latter is given by \cite{Birk,Birkbis}
\[
\widetilde{S}\Psi(z)=\sqrt{\det tS}\Psi(Sz)=\Psi(Sz)
\]
since $\sqrt{\det tS}=1$ because $S\in\operatorname*{Sp}(n)$.
\end{proof}
\end{proposition}

\subsection{The intertwiners $U_{\phi}$ and their properties}

For $\phi\in L^{2}(\mathbb{R}^{n}\mathbb{)}$ we define the mapping $U_{\phi
}:L^{2}(\mathbb{R}^{n}\mathbb{)\longrightarrow}L^{2}(\mathbb{R}^{2n}%
\mathbb{)}$ by%
\begin{equation}
U_{\phi}\psi=(2\pi\hbar)^{n/2}W(\psi,\phi) \label{inter}%
\end{equation}
where $W(\psi,\phi)$ is the cross-Wigner transform (\ref{crosswig}). We will
call $U_{\phi}$ the \emph{wavepacket transform with window} $\phi$. It follows
from Moyal's identity (\cite{WIGNER}, Prop. 65)%
\begin{equation}
(W(\psi,\phi)|W(\psi^{\prime},\phi^{\prime}))_{L^{2}(\mathbb{R}^{2n}%
\mathbb{)}}=\left(  \tfrac{1}{2\pi\hbar}\right)  ^{n}(\psi|\psi^{\prime
})_{L^{2}(\mathbb{R}^{n}\mathbb{)}}(\phi|\phi^{\prime})_{L^{2}(\mathbb{R}%
^{n}\mathbb{)}} \label{Moyal}%
\end{equation}
that $U_{\phi}$ is a partial isometry of $L^{2}(\mathbb{R}^{n}\mathbb{)}$ onto
a closed subspace $\mathcal{H}_{\phi}$ of $L^{2}(\mathbb{R}^{2n}\mathbb{)}$.
It follows that $U_{\phi}^{\ast}U_{\phi}$ is the identity on $L^{2}%
(\mathbb{R}^{n}\mathbb{)}$ and that $U_{\phi}U_{\phi}^{\ast}$ is the
orthogonal projection onto $\mathcal{H}_{\phi}$. We call the mappings
$U_{\phi}$ \textquotedblleft windowed wavepacket transforms\textquotedblright,
their interest in the context of Bopp calculus comes from the following
important intertwining result:

\begin{proposition}
\label{prop2}Let $\widetilde{A}=\operatorname*{Op}_{\mathrm{Bopp}}(a)$ and
$\widehat{A}=\operatorname*{Op}_{\mathrm{Weyl}}(a)$ . We have
\begin{equation}
\widetilde{A}U_{\phi}=U_{\phi}\widehat{A}\text{ \ and \ }U_{\phi}^{\ast
}\widetilde{A}=\widehat{A}U_{\phi}^{\ast} \label{UA1}%
\end{equation}
and hence
\begin{equation}
\widehat{A}=U_{\phi}^{\ast}\widetilde{A}U_{\phi}\text{ \ and \ }_{\phi}^{\ast
}\widetilde{A}=U_{\phi}\widetilde{A}U_{\phi}^{\ast}\text{ }. \label{UA2}%
\end{equation}

\end{proposition}

\begin{proof}
Using the harmonic representation (\ref{aharmonic}) of $\widetilde{A}$ \ the
formulas (\ref{UA1}) follows from the fact that
\[
\widetilde{T}(z_{0})U_{\phi}\psi=)U_{\phi}(\widehat{T}(z_{0}\psi)
\]
(see \cite{Birkbis,WIGNER} for details). Formulas (\ref{UA2}) immediately
follow using the relations $U_{\phi}^{\ast}U_{\phi}=I_{d}$.
\end{proof}

Notice that the Hilbert spaces $\mathcal{H}_{\phi}$ do not cover
$L^{2}(\mathbb{R}^{2n}\mathbb{)}$. There exist (infinitely many) $\Psi\in
L^{2}(\mathbb{R}^{2n}\mathbb{)}$ such that there do not exist $(\psi,\phi)\in
L^{2}(\mathbb{R}^{n}\mathbb{)}$ $\times L^{2}(\mathbb{R}^{n}\mathbb{)}$ \ such
that $\Psi=U_{\phi}\psi.$ The archetypical examples are given by Gaussians
which are too concentrated around their center because such functions violate
Hardy's uncertainty principle \cite{Hardy}. However:

\begin{proposition}
\label{Propbase}Let $(\psi_{j})_{j}$ and $(\phi_{k})_{k}$ be orthonormal bases
of $L^{2}(\mathbb{R}^{n}\mathbb{)}$. Then $(U_{\phi_{k}}\psi_{j})_{jk}$ is an
orthonormal basis of $L^{2}(\mathbb{R}^{2n}\mathbb{)}$.
\end{proposition}

\begin{proof}
(This is a simplified version of the proof of Thm. 441, \S 19.1 in
\cite{Birkbis}.) That $(U_{\phi_{k}}\psi)_{jk}$ is an orthonormal system in
$L^{2}(\mathbb{R}^{2n}\mathbb{)}$ is clear from Moyal's identity
(\ref{Moyal}). To show that the vectors $(U_{\phi_{k}}\psi_{j})_{jk}$ span
$L^{2}(\mathbb{R}^{2n}\mathbb{)}$ it suffices to prove that if $\Psi\in
L^{2}(\mathbb{R}^{2n}\mathbb{)}$ satisfies $(\Psi|W(\psi_{j},\phi_{k}))=0$ for
all $j,k$ \ then $\Psi=0$. Let us view $\Psi$ as the Weyl symbol of an
operator $\widehat{A}_{\Psi}$ on $L^{2}(\mathbb{R}^{n}\mathbb{)}$. We have
\[
(\Psi|W(\psi_{j},\phi_{k}))=(\widehat{A}_{\Psi}\psi_{j}|\phi_{k})
\]
(formula (10.8) in \cite{Birkbis}) hence the relations $(\Psi|W(\psi_{j}%
,\phi_{k}))=0$ are equivalent to
\[
(\widehat{A}_{\Psi}\psi_{j}|\phi_{k})=0\text{ \ for all }j,k.
\]
Since $(\phi_{k})_{k}$ is a basis of $L^{2}(\mathbb{R}^{n}\mathbb{)}$ this
implies $\widehat{A}_{\Psi}\psi_{j}=0$ for all $j$ hence $\widehat{A}_{\Psi
}=0$ since $(\psi_{j})_{j}$ is a basis. Since there is one-to-one
correspondence between Weyl operators and their symbols we must have $\Psi=0$.
\end{proof}

We have:

\begin{proposition}
(i) Let $\widehat{A}=\operatorname*{Op}_{\mathrm{Weyl}}(a)$ be a compact
operator on $L^{2}(\mathbb{R}^{n}\mathbb{)}$ and $\widetilde{A}%
=\operatorname*{Op}_{\mathrm{Bopp}}(a)$ the corresponding Bopp operator. \ (i)
$\widetilde{A}$. is self-adjoint if and only $\widehat{A}$ is (ii)
$\widetilde{A}$; (ii) $\widehat{A}$ and $\widetilde{A}$ have the same
eigenvalues. (ii) If $\psi$ is an eigenvector of $\widehat{A}$ corresponding
to the eigenvalue $\lambda$ then each function $\Psi=U_{\phi}\psi$ \ is an
eigenfunctions of $\widetilde{A}$ for the eigenvalue\ $\lambda$.
\end{proposition}

\begin{proof}
(i) $\widetilde{A}$ is self adjoint if and only its Weyl symbol $(z,\zeta
)\longmapsto a(z-\frac{1}{2}J\zeta)$ is a real function, but this is possible
if and only if the symbol $a$ is real, that is if $\widehat{A}$ is
self-adjoint. (ii) and (iii) Assume that $\widehat{A}\psi=\lambda\psi$ for
$\psi\in L^{2}(\mathbb{R}^{n}\mathbb{)}$, $\psi\neq0$. Then $\widetilde{A}%
(U_{\phi}\psi)=U_{\phi}(\widehat{A}\psi)=\lambda(U_{\phi}\psi)$ hence
$\lambda$ is an eigenvalue since $U_{\phi}\psi\neq0$ because $U_{\phi}$ is
injective. Suppose conversely that $\widetilde{A}\Psi=\lambda\Psi$ for
$\Psi\in L^{2}(\mathbb{R}^{2n}\mathbb{)}$, $\Psi\neq0$. Then
\[
U_{\phi}^{\ast}(\widetilde{A}\Psi)=\widehat{A}(U_{\phi}^{\ast}\Psi)=\lambda
U_{\phi}^{\ast}\Psi
\]
so $\lambda$ is an eigenvalue of $\widehat{A}$ if $U_{\phi}^{\ast}\Psi\neq0$
for at least one window $\phi$. Suppose $U_{\phi}^{\ast}\Psi=0$ for every
$\phi$; then for every $\theta\in L^{2}(\mathbb{R}^{n}\mathbb{)}$
\[
(U_{\phi}^{\ast}\Psi|\theta)_{L^{2}(\mathbb{R}^{n}\mathbb{)}}=(U_{\phi}^{\ast
}\Psi|\theta)_{L^{2}(\mathbb{R}^{n}\mathbb{)}}=(\Psi|U_{\phi}\varkappa
)_{L^{2}(\mathbb{R}^{n}\mathbb{)}}=0
\]
but this implies $\Psi=0$ in view of Proposition \ref{Propbase} above: \ let
$(\phi_{j})$ and $(\varkappa_{k})$ be orthonormal bases of $L^{2}%
(\mathbb{R}^{2n}\mathbb{)}$, then $(U_{\phi_{j}}\varkappa_{k})$ is an
orthonormal basis of $L^{2}(\mathbb{R}^{2n}\mathbb{)}$ and the relations
$(\Psi|U_{\phi_{j}}\varkappa_{k})_{L^{2}(\mathbb{R}^{n}\mathbb{)}}=0$ for all
$j,k$ imply $\Psi=0$ which is a contradiction.
\end{proof}

We mention that the \ wavepacket transforms $U_{\phi}$ are closely related to
the Gabor (or short-time Fourier) rainstorms
\begin{equation}
V_{\phi}\psi(z)=\int e^{-2\pi ipx}\psi(x^{\prime})\phi^{\ast}(x^{\prime
}-x)dx^{\prime} \label{STFT}%
\end{equation}
used in time-frequency analysis and signal theory \cite{Gro}. In the case
$\hbar=1/2\pi$ we have
\[
U_{\phi}(z)=\left(  \frac{\pi\hbar}{2}\right)  ^{n/2}e^{-i\pi px\pi}%
W(\psi,\phi^{-})(\frac{1}{2}z)
\]
where $\phi^{-}(x)=\phi(-x)$.

\section{The Care of Density Operators}

\subsection{Density operators on\ $L^{2}(\mathbb{R}^{n}\mathbb{)}$}

For a detailed study and proofs see Chapter 14 in \cite{QHA}. A density
operator (also called density matrix) $\widehat{\rho}$ on $L^{2}%
(\mathbb{R}^{n}\mathbb{)}$ is a positive semidefinite operator $\widehat{\rho
}$ $\geq0$ with trace one: $\operatorname*{Tr}\widehat{\rho}=1$. It is thus a
self-adjoint compact operator on $L^{2}(\mathbb{R}^{n}\mathbb{)}$. In physics
density operators are identified with mixed quantum states, that is classical
mixtures $(\psi_{j},\alpha_{j})$, $\psi_{j}\in L^{2}(\mathbb{R}^{n}\mathbb{)}%
$, $\alpha_{j}\geq0$, $\sum_{j}\alpha_{j}=1$ of pure (normed) states Such a
state is identified with the operator $\widehat{\rho}=\sum_{j}\alpha_{j}%
\Pi_{\psi_{j}}$. One shows (\cite{QHA}, \S 1.4.) that $\widehat{\rho}$ is a
density operator. In fact, it follows from the spectral theorem for compact
self-adjoint operators, that, conversely, always exist a sequence of
orthonormal functions $(\psi_{j})$ in $L^{2}(\mathbb{R}^{n}\mathbb{)}$ and a
sequence of real numbers $(\lambda_{j})$ such that
\begin{equation}
\widehat{\rho}=\sum_{j}\lambda_{j}\widehat{\Pi}_{\psi_{j}}\text{ \ },\text{
\ }\lambda_{j}\geq0\text{ \ },\text{ \ }\sum_{j}\lambda_{j}=1 \label{proj}%
\end{equation}
where $\widehat{\Pi}_{\psi_{j}}$ is the orthogonal projection onto the ray
spanned by \ $\psi_{j}$: $\widehat{\Pi}_{\psi_{j}}\psi=(\psi|\psi_{j})\psi
_{j}$. The \ numbers $\lambda_{j}$ are the eigenvalues of $\widehat{\rho}$ and
are of finite multiplicity and $\lim_{j\rightarrow\infty}\lambda_{j}=0$. The
$\psi_{j}$ are the corresponding eigenfunctions. Since $\widehat{\rho}$ \ is a
continuous mapping $L^{2}(\mathbb{R}^{n}\mathbb{)\longrightarrow}%
L^{2}(\mathbb{R}^{n}\mathbb{)}$ it follows from Schwartz's kernel theorem that
$\widehat{\rho}$ is a Weyl operator, it symbol is $(2\pi\hbar)^{n}\rho$ where
$\rho$ (the Wigner distribution of $\widehat{\rho}$) is the function defined
by%
\begin{equation}
\rho(z)=\sum_{j}\lambda_{j}W\psi_{j}(z) \label{Wigdis}%
\end{equation}
where $W\psi_{j}=W(\psi_{j},\psi_{j})$ is the usual Wigner transform of
$\psi_{j}$.\ If $\rho\in L^{1}(\mathbb{R}^{2n})$ then $\rho$ is a quasi
probability distribution in the sense that
\begin{equation}
\operatorname*{Tr}\widehat{\rho}=\int_{\mathbb{R}^{2n}}\rho(z)dz=1.
\label{tr1}%
\end{equation}
When $\psi_{j}$ and its Fourier transform $\widehat{\psi}$ are in
$L^{1}(\mathbb{R}^{n}\mathbb{)\cap}L^{2}(\mathbb{R}^{n}\mathbb{)}$ the
marginal properties%
\begin{align}
\int_{\mathbb{R}^{n}}\rho(x,p)dp  &  =\sum_{j}\lambda_{j}|\psi_{j}%
(x)|^{2}\label{mrg1}\\
\int_{\mathbb{R}^{n}}\rho(x,p)dx  &  =\sum_{j}\lambda_{j}|\widehat{\psi}%
_{j}(p)|^{2} \label{marg2}%
\end{align}

hold (see \cite{Folland,Birkbis,WIGNER}).

\subsection{Bopp quantization of the density operator}

We consider a density operator $\widehat{\rho}=\operatorname*{Op}%
_{\mathrm{Weyl}}((2\pi\hbar)^{n}\rho)$ on $L^{2}(\mathbb{R}^{n}\mathbb{)}$ and
the corresponding Bopp operator $\widetilde{\rho}=\operatorname*{Op}%
_{\mathrm{Bopp}}((2\pi\hbar)^{n}\rho)$.  These operators have the same
eigenvalues, but $\widetilde{\rho}$ cannot be a density a density operator;
this is immediately visible if one recalls that the Weyl symbol $\widetilde{a}%
$ of $\widetilde{\rho}$ is $\widetilde{\rho}(z,\varsigma)=\rho(z-\frac{1}%
{2}J\varsigma)$ and can thus not have finite integral as expected from a
quasi-distribution; in fact, in general,%
\begin{equation}
\int_{\mathbb{R}^{4n}}\widetilde{\rho}(z,\varsigma)=\left(  \int%
_{\mathbb{R}^{2n}}\int_{\mathbb{R}^{2n}}\rho(z-\frac{1}{2}J\varsigma
)dz\right)  d\varsigma=\infty.\label{rev6}%
\end{equation}
However:

\begin{proposition}
For every window $\phi$ the restriction $\widetilde{\rho}_{\phi}$ of
$\widetilde{\rho}$ to the Hilbert space $\mathcal{H}_{\phi}$ is a density
operator on $\mathcal{H}_{\phi}$ given by
\begin{equation}
\widetilde{\rho}_{\phi}=\sum_{j}\lambda_{j}\widetilde{\Pi}_{\Psi_{j}}\text{
\ when }\widehat{\rho}=\sum_{j}\lambda\widehat{\Pi}_{\psi_{_{j}}}
\label{roffi}%
\end{equation}
where the $\lambda$ ar4e the eigenvalues of $\widehat{\rho}$ and
$\widetilde{\Pi}_{\Psi_{j}}$ is the orthogonal projection in $\mathcal{H}%
_{\phi}$ on on the eigenfunctions $\Psi_{j}=U_{\phi}\psi_{j}$.
\end{proposition}

\begin{proof}
Let $\Psi\in\mathcal{H}_{\phi}$: we have $\Psi=U_{\phi}\psi$ \ for some
$\psi\in L^{2}(\mathbb{R}^{n}\mathbb{)}$ hence
\[
\widetilde{\rho}_{\phi}\Psi=\widetilde{\rho}U_{\phi}\psi=U_{\phi
}(\widehat{\rho}\psi).
\]
The operator $\widetilde{\rho}_{\phi}$ is positive semidefinite (and hence
self-adjoint) for we have
\begin{align*}
(\widetilde{\rho}_{\phi}\Psi|\Psi)_{L^{2}(\mathbb{R}^{2n}\mathbb{)}}  &
=(U_{\phi}(\widehat{\rho}\psi)|U_{\phi}\psi)_{L^{2}(\mathbb{R}^{2n}\mathbb{)}%
}\\
&  =(\widehat{\rho}\psi|\psi)_{L^{2}(\mathbb{R}^{n}\mathbb{)}}\geq0.
\end{align*}
Let us show that $\widetilde{\rho}_{\phi}$ is of trace class, for this it
suffices to show that if $((\Psi_{j}),(\Phi_{k}))$ is a pair of orthonormal
basis of $\mathcal{H}_{\phi}$ then the series $\sum_{j,k}(\widetilde{\rho
}_{\phi}\Psi_{j}|\Phi_{k})_{L^{2}(\mathbb{R}^{2n}\mathbb{)}}$ is absolutely
convergent. We have $\Psi_{j}=U_{\phi}\psi_{j}$ and $\Phi_{k}=U\phi_{k}$
\ where $\psi_{j})$ and $(\phi_{k})$ are orthonormal bases of $L^{2}%
(\mathbb{R}^{n}\mathbb{)}$ \ and the claim follows from the equality
\begin{align*}
\sum_{j,k}|(\widetilde{\rho}_{\phi}\Psi_{j}|\Phi_{k})_{L^{2}(\mathbb{R}%
^{2n}\mathbb{)}}|  &  =\sum_{j,k}|(U_{\phi}(\widehat{\rho}\psi_{j})|U_{\phi
}(\widehat{\rho}\phi_{k}))_{L^{2}(\mathbb{R}^{2n}\mathbb{)}}|\\
&  =\sum_{j,k}|(\widehat{\rho}\psi_{j}|\widehat{\rho}\phi_{k}))_{L^{2}%
(\mathbb{R}^{n}\mathbb{)}}|<\infty
\end{align*}
since $\widehat{\rho}$ is of trace class. Let $(\Psi_{j})=(U_{\phi}$ $\psi
_{j}($ be an orthonormal basis of $\mathcal{H}_{\phi}$; the trace of
$\widetilde{\rho}_{\phi}$ is%
\begin{align*}
\operatorname*{Tr}\widetilde{\rho}_{\phi}  &  =\sum_{j}|(\widetilde{\rho
}_{\phi}\Psi_{j}|\Psi_{j}))_{L^{2}(\mathbb{R}^{2n}\mathbb{)}})=\sum
_{j}(U_{\phi}(\widehat{\rho}\psi_{j})|U_{\phi}\psi_{j}))_{L^{2}(\mathbb{R}%
^{2n}\mathbb{)}})\\
&  =\sum_{j}(\widehat{\rho}\psi_{j})|\psi_{j}))_{L^{2}(\mathbb{R}%
^{n}\mathbb{)}}=\operatorname*{Tr}\widehat{\rho}=1.
\end{align*}
The proof of formula (\ref{roffi}) goes as follows. assume that $\widehat{\rho
}=\sum_{j}\lambda\widehat{\Pi}_{\psi_{_{j}}}$. Then%
\begin{align*}
U_{\phi}(\widehat{\Pi}_{\psi_{j}}\psi)  &  =(2\pi\hbar)^{n/2}W(\widehat{\Pi
}_{\psi_{j}}\psi,\phi)\\
&  =(2\pi\hbar)^{n/2}(\psi|\psi_{j})_{L^{2}(\mathbb{R}^{n}\mathbb{)}}%
W(\psi_{j},\phi)\\
&  =(\psi|\psi_{j})_{L^{2}(\mathbb{R}^{n}\mathbb{)}}U_{\phi}\psi_{j}.
\end{align*}
Since $U_{\phi}$ is a partial isometry we have \ $(\psi|\psi_{j}%
)_{L^{2}(\mathbb{R}^{n}\mathbb{)}}=(U_{\phi}\psi|U_{\phi}\psi_{j}%
)_{L^{2}(\mathbb{R}^{2n}\mathbb{)}}$ hence%
\[
U_{\phi}(\widehat{\Pi}_{\psi_{j}}\psi)=(\Psi|\Psi_{j})_{L^{2}(\mathbb{R}%
^{2n}\mathbb{)}}\Psi
\]
which proves (\ref{roffi}).
\end{proof}

\section{Deformation Quantization of the Density Operator}

Deformation quantization is a mathematical framework that allows the
transition from classical to quantum mechanics by "deforming" the classical
algebra of observables into a non-commutative algebra, reflecting the quantum
nature of observables. The Moyal product is the simplest example of
deformation quantization and serves as a prime example of this approach.

One starts with the algebra $C^{\infty}(\mathbb{R}^{2n},\mathbb{R})$ of
classical observables endowed with the Poisson bracket $P=\{\cdot,\cdot\}$:%
\begin{equation}
P(a,b)=\{a,b\}=\sum_{j}\frac{\partial a}{\partial x_{j}}\frac{\partial
b}{\partial p_{j}}-\frac{\partial a}{\partial p_{j}}\frac{\partial b}{\partial
x_{j}}.\label{rev7}%
\end{equation}
It can be expressed in terms of the standard symplectic matrix $J=(J_{\alpha
\beta})_{1\leq\alpha,\beta\leq n}=%
\begin{pmatrix}
0 & I\\
-I & 0
\end{pmatrix}
$ as%
\begin{equation}
P(a,b)=\sum_{1\leq\alpha,\beta\leq n}J_{\alpha\beta}\frac{\partial a}{\partial
z_{\alpha}}\frac{\partial b}{\partial z_{\beta}}\label{rev8}%
\end{equation}
where we have set $z_{\alpha}=x_{\alpha}$ if $1\leq\alpha\leq n$ and
$z_{\alpha}=p_{\alpha}$ if if $n+1\leq\alpha\leq2n$. With this notation the
Poisson bracket can be written concisely in terms of the symplectic form as
\begin{equation}
P(a,b)=\sigma(\partial_{z}a,\partial_{z}b)\label{pdzab}%
\end{equation}
where $\partial_{z}$ is the gradient in the variolas $z_{1},...,z_{2n}$. The
algebra $C^{\infty}(\mathbb{R}^{2n},\mathbb{R}),\{\cdot,\cdot\}$ is an
infinitely dimensional Lie algebra.

In deformation quantization one uses the vector space $C^{\infty}%
(\mathbb{R}^{2n},\mathbb{R})[\hbar]$ of formal series $a(\hbar)=\sum_{r}%
\hbar^{r}a_{r}$ where the $a_{r}$ \ are classical observables equipped with
the obvious vector space laws; it becomes a commutative and associative
algebra when endowed with the multiplicative law%
\[
a(\hbar)b(\hbar)=\sum_{r\geq0}\hbar^{r}\left(  \sum_{i+j=r}aib_{j}\right)  .
\]
With this multiplicative law $C^{\infty}(\mathbb{R}^{2n},\mathbb{R})[\hbar]$
becomes an algebra, the \emph{algebra of \ quantum observables}. The following
result makes the connection between formal series and the Moyal star product
(Groenewold \cite{Groenewold}, Bayen \textit{et al}. \cite{Bayen}, Voros
\cite{Voros2}. ):

\begin{proposition}
\label{Propstar}The Moyal product $a\bigstar_{\hbar}b$ is identified with the
formal series%
\begin{equation}
a\bigstar_{\hbar}b=\sum_{r\geq0}\frac{\hbar^{r}}{r!}P^{r}(a,b)
\label{starmoyal}%
\end{equation}
where the operator $P^{r}$ is the $r$-th power of the Poisson bracket, defined
by
\begin{equation}
P^{r}(a,b)(z)=\left.  \sum_{i=1}^{r}\left(  \frac{\partial^{2}}{\partial
x_{i}\partial p_{i}^{\prime}}-\frac{\partial^{2}}{\partial x_{i}^{\prime
}\partial p_{i}}\right)  ^{r}a(z)b(z^{\prime})\right\vert _{z=z^{\prime}}.
\label{pr}%
\end{equation}

\end{proposition}

In particular
\[
a\bigstar_{\hbar}b=ab+\hbar\{a,b\}\sum_{r\geq2}\frac{\hbar^{r}}{r!}%
P^{r}(a,b).
\]

Combining the definitions above with the study of the Bopp operators we have
the following characterization of deformation quantization of the density operator:

\begin{proposition}
Assume that the density operator $\widehat{\rho}$ on $L^{2}(\mathbb{R}%
^{n}\mathbb{)}$ is given by (\ref{proj}).Let $\widetilde{\rho}_{\phi}$ the
associated Bopp phase space density operator for some window $\phi$. Let
$\Psi=U_{\phi}\psi$, $\psi\in$ $L^{2}(\mathbb{R}^{n}\mathbb{)}$. We have%
\begin{equation}
\widetilde{\rho}_{\phi}\Psi=(2\pi\hbar)^{3n/2}\sum_{r\geq0}\frac{\hbar^{r}%
}{r!}P^{r}(\rho,W(\psi,\phi) \label{1}%
\end{equation}
where $\rho$ is the Wigner distribution of $\widehat{\rho}$, explicitly%
\begin{equation}
P^{r}(\rho,W(\psi,\phi)=\sum_{j}\lambda_{j}P^{r}(W\psi_{j},W(\psi,\phi).
\label{2}%
\end{equation}

\end{proposition}

\begin{proof}
We gave%
\[
\widehat{\rho}=\sum_{j}\lambda_{j}\widehat{\Pi}_{\psi_{j}}\ ,\ \lambda_{j}%
\geq0\ ,\ \sum_{j}\lambda_{j}=1
\]
and hence $\rho(z)=\sum_{j}\lambda_{j}W\psi_{j}(z)$. In view of the
definitions of the intertwiners $U_{\phi}$ (\ref{inter}) and of the Bopp
operators \ref{defBopp}) we have
\begin{align*}
\widetilde{\rho}_{\phi}\Psi &  =\widetilde{\rho}(U_{\phi}\psi)=(2\pi
\hbar)^{n/2}\rho\bigstar_{\hbar}W(\psi,\phi)\\
&  =(2\pi\hbar)^{n/2}\sum_{j}\lambda_{j}W\psi_{j}\bigstar_{\hbar}W(\psi,\phi).
\end{align*}
Using formulas (\ref{starmoyal}) and (\ref{pr}) above have%
\[
\lambda_{j}W\psi_{j}\bigstar_{\hbar}W(\psi,\phi)=\lambda_{j}\sum_{r\geq0}%
\frac{\hbar^{r}}{r!}P^{r}(W\psi_{j},W(\psi,\phi))
\]
hence the result summing over the indices $j$.
\end{proof}

\section{Comments and Perspectives}

We have studied an elementary case of deformation quantization (the Moyal star
product) within the framework of the Weyl formalism, which is the most used
quantization procedure used by physicists. This popularity is due to many
factors, some being historical, \ others being due to its relative simplicity.
But the main advantage of Weyl quantization comes from that it is the only
pseudodifferential calculus satisfying the property of symplectic covariance:
It would however be interesting to adapt this theory to other possible
quantizations, in particular the Born--Jordan procedure we have developed
rigorously in previous work (see for instance our per \cite{Boppgolu2} with
Luef). In his deep but not so well-known work \cite{Leray} Leray introduces a
class of functions on a Lagrangian submanifold $V$ of $(\mathbb{R}^{2n}%
,\sigma)$. These functions, which Leray calls \textquotedblleft Lagrangian
functions\textquotedblright\ are formal series defined on the universal
covering $\check{V}$ of $V$, and are characterized by their phase $\varphi$,
defined by $d\varphi=pdx$ on $\check{V}$. These Lagrangian functions become
functions on $V$ if one imposes a certain quantization condition on this
manifold. It turns out that Leray's Lagrangian functions appear to be more
general objects than the deformation quantization formal series
(\ref{starmoyal}), to which they reduce when $V$ is flat. It would certainly
be interesting to study deformation quantization from the perspective of
Leray's constructions. Bopp pseudodifferential operators could play an
important role in this perspective.

We mention that Wong \cite{Wong1} has used related methods to analyze the
global hypoellipticity of operators. In a recent paper \cite{ElenaLuigi}
Cordero and Rodino study closely related topics (also see Cordero \textit{et
al}. \cite{co25} in that context).

\begin{acknowledgement}
This work has been financed by the Austrian Research Foundation FWF (Grant
number PAT 2056623). It was done during a stay of the author at the Acoustics
Research Institute group at the Austrian Academy of Sciences. We thank the
director of the Acoustics Research Institute, Professor Peter Balasz, for his
kind hospitality. I thank Guiseppe Dito for having provided me with some
useful references about deformation quantization. Also many thanks to Hans
Feichtinger and Luigi Rodino for useful commentsand addtional refernces.
\end{acknowledgement}

MauriceAlexis.deGossondeVarennes@oeaw.ac.at

maurice.de.gosson@univie.ac.at
\end{document}